\documentclass[pra,preprint,groupedaddress,showmsc,superscriptaddress,twocolumn,10pt]{revtex4}
\usepackage{amsmath,amsfonts,amssymb,caption,hyperref,color,epsfig,graphics,graphicx,latexsym,mathrsfs,revsymb,theorem,url,verbatim,epstopdf,cleveref,dashrule,cases,booktabs,diagbox,threeparttable,setspace,ragged2e}
\hypersetup{colorlinks,linkcolor={blue},citecolor={blue},urlcolor={red}}  

\usepackage[center]{titlesec}
\hypersetup{colorlinks,linkcolor={blue},citecolor={red},urlcolor={blue}}

\newtheorem{definition}{Definition}

\newtheorem{lemma}[definition]{Lemma}

\newtheorem{theorem}[definition]{Theorem}

\newtheorem{example}[definition]{Example}

\def\squareforqed{\hbox{\rlap{$\sqcap$}$\sqcup$}}
\def\qed{\ifmmode\squareforqed\else{\unskip\nobreak\hfil
\penalty50\hskip1em\null\nobreak\hfil\squareforqed
\parfillskip=0pt\finalhyphendemerits=0\endgraf}\fi}
\def\endenv{\ifmmode\;\else{\unskip\nobreak\hfil
\penalty50\hskip1em\null\nobreak\hfil\;
\parfillskip=0pt\finalhyphendemerits=0\endgraf}\fi}
\newenvironment{proof}{\noindent \textbf{{Proof.~} }}{\qed}
\def\Dbar{\leavevmode\lower.6ex\hbox to 0pt
{\hskip-.23ex\accent"16\hss}D}
\def\bpf{\begin{proof}}
\def\epf{\end{proof}}

\newcommand{\bra}[1]{\langle{#1}|}
\newcommand{\ket}[1]{|{#1}\rangle}
\newcommand{\proj}[1]{|{#1}\rangle \langle {#1}|}
\newcommand{\ketbra}[2]{|{#1}\rangle \! \langle{#2}|}

\newcommand{\abs}[1]{\left\lvert {#1} \right\rvert}

\newcommand{\etal}{{\sl et~al.}}

\newcommand{\red}{\textcolor{red}}

\newcommand{\opp}{\red{OPEN PROBLEMS}.~}
\newcommand{\tbc}{\red{TO BE CONTINUED}.~}

\newcommand{\nc}{\newcommand}

\def\bea{\begin{eqnarray}}
\def\eea{\end{eqnarray}}
\def\beq{\begin{equation}}
\def\eeq{\end{equation}}
\def\bal{\begin{aligned}}
\def\eal{\end{aligned}}
\def\bma{\begin{bmatrix}}
\def\ema{\end{bmatrix}}

\def\rank{\mathop{\rm rank}}

\def\diag{\mathop{\rm diag}}
\def\tr{{\rm Tr}}
\def\dim{\mathop{\rm Dim}}

\def\dg{\dagger}

\def\ox{\otimes}

\def\a{\alpha}
\def\b{\beta}

\def\m{\mu}
\def\n{\nu}

\def\r{\rho}

\def\ph{\varphi}

\def\ps{\psi}
\def\o{\omega}

\def\Ps{\Psi}
\def\O{\Omega}

\nc{\bbA}{\mathbb{A}} \nc{\bbB}{\mathbb{B}} \nc{\bbC}{\mathbb{C}}
\nc{\bbD}{\mathbb{D}} \nc{\bbE}{\mathbb{E}} \nc{\bbF}{\mathbb{F}}
\nc{\bbG}{\mathbb{G}} \nc{\bbH}{\mathbb{H}} \nc{\bbI}{\mathbb{I}}
\nc{\bbJ}{\mathbb{J}} \nc{\bbK}{\mathbb{K}} \nc{\bbL}{\mathbb{L}}
\nc{\bbM}{\mathbb{M}} \nc{\bbN}{\mathbb{N}} \nc{\bbO}{\mathbb{O}}
\nc{\bbP}{\mathbb{P}} \nc{\bbQ}{\mathbb{Q}} \nc{\bbR}{\mathbb{R}}
\nc{\bbS}{\mathbb{S}} \nc{\bbT}{\mathbb{T}} \nc{\bbU}{\mathbb{U}}
\nc{\bbV}{\mathbb{V}} \nc{\bbW}{\mathbb{W}} \nc{\bbX}{\mathbb{X}}
\nc{\bbZ}{\mathbb{Z}}

\nc{\bA}{{\bf A}} \nc{\bB}{{\bf B}} \nc{\bC}{{\bf C}}
\nc{\bD}{{\bf D}} \nc{\bE}{{\bf E}} \nc{\bF}{{\bf F}}
\nc{\bG}{{\bf G}} \nc{\bH}{{\bf H}} \nc{\bI}{{\bf I}}
\nc{\bJ}{{\bf J}} \nc{\bK}{{\bf K}} \nc{\bL}{{\bf L}}
\nc{\bM}{{\bf M}} \nc{\bN}{{\bf N}} \nc{\bO}{{\bf O}}
\nc{\bP}{{\bf P}} \nc{\bQ}{{\bf Q}} \nc{\bR}{{\bf R}}
\nc{\bS}{{\bf S}} \nc{\bT}{{\bf T}} \nc{\bU}{{\bf U}}
\nc{\bV}{{\bf V}} \nc{\bW}{{\bf W}} \nc{\bX}{{\bf X}}
\nc{\bZ}{{\bf Z}}

\nc{\bmA}{{\bm A}} \nc{\bmB}{{\bm B}} \nc{\bmC}{{\bm C}}
\nc{\bmD}{{\bm D}} \nc{\bmE}{{\bm E}} \nc{\bmF}{{\bm F}}
\nc{\bmG}{{\bm G}} \nc{\bmH}{{\bm H}} \nc{\bmI}{{\bm I}}
\nc{\bmJ}{{\bm J}} \nc{\bmK}{{\bm K}} \nc{\bmL}{{\bm L}}
\nc{\bmM}{{\bm M}} \nc{\bmN}{{\bm N}} \nc{\bmO}{{\bm O}}
\nc{\bmP}{{\bm P}} \nc{\bmQ}{{\bm Q}} \nc{\bmR}{{\bm R}}
\nc{\bmS}{{\bm S}} \nc{\bmT}{{\bm T}} \nc{\bmU}{{\bm U}}
\nc{\bmV}{{\bm V}} \nc{\bmW}{{\bm W}} \nc{\bmX}{{\bm X}}
\nc{\bmZ}{{\bm Z}}

\nc{\cA}{{\cal A}} \nc{\cB}{{\cal B}} \nc{\cC}{{\cal C}}
\nc{\cD}{{\cal D}} \nc{\cE}{{\cal E}} \nc{\cF}{{\cal F}}
\nc{\cG}{{\cal G}} \nc{\cH}{{\cal H}} \nc{\cI}{{\cal I}}
\nc{\cJ}{{\cal J}} \nc{\cK}{{\cal K}} \nc{\cL}{{\cal L}}
\nc{\cM}{{\cal M}} \nc{\cN}{{\cal N}} \nc{\cO}{{\cal O}}
\nc{\cP}{{\cal P}} \nc{\cQ}{{\cal Q}} \nc{\cR}{{\cal R}}
\nc{\cS}{{\cal S}} \nc{\cT}{{\cal T}} \nc{\cU}{{\cal U}}
\nc{\cV}{{\cal V}} \nc{\cW}{{\cal W}} \nc{\cX}{{\cal X}}
\nc{\cZ}{{\cal Z}}

\usepackage{dsfont}

\begin{document}

\title{Absolutely maximally entangled states in tripartite heterogeneous systems}

\author{Yi Shen}\email[]{yishen@buaa.edu.cn}
\affiliation{School of Mathematical Sciences, Beihang University, Beijing 100191, China}
\author{Lin Chen}\email[]{linchen@buaa.edu.cn (corresponding author)}
\affiliation{School of Mathematical Sciences, Beihang University, Beijing 100191, China}
\affiliation{International Research Institute for Multidisciplinary Science, Beihang University, Beijing 100191, China}

\begin{abstract}
Absolutely maximally entangled (AME) states are typically defined in homogeneous systems. 
However, the quantum system is more likely to be heterogeneous in a practical setup.
In this work we pay attention to the construction of AME states in tripartite heterogeneous systems. 
We first introduce the concept of irreducible AME states as the basic elements to construct AME states with high local dimensions.
Then we investigate the tripartite heterogeneous systems whose local dimensions are $l,m,n$, with $3\leq l<m<n\leq m+l-1$. We show the existence of AME states in such heterogeneous systems is related to a kind of arrays called magic solution array. We further identify the AME states in which kinds of heterogeneous systems are irreducible. In addition, we propose a method to construct $k$-uniform states of more parties using two existing AME states. We also build the connection between heterogeneous AME states and multi-isometry matrices, and indicate an application in quantum steering.
\end{abstract}

\date{\today}

\maketitle



\section{Introduction}
\label{sec:intro}

Quantum correlations play a central role in the foundation of quantum mechanics, kinds of quantum information processing, and the physics of strongly correlated systems \cite{marginallin14}. The nonclassicality of quantum correlations, in particular entanglement, challenges our understanding of the relation between local and global properties of quantum states.
For instance, the Bell state is a pure two-qubit state while each of its reduced density matrices is maximally mixed. It implies that even if we have the complete knowledge of the global state, the entropy of the reduced state can be the maximum. 
There is a class of pure states called absolutely maximally entangled (AME) states which share the similar property with the Bell state. Denote by $\cH_d$ a $d$-dimensional Hilbert space. A system is homogeneous if its local dimensions are equal, e.g., $\cH_d^{\ox n}$. Suppose $\ket{\phi}\in\cH_d^{\ox n}$is an $n$-partite pure state. We call $\ket{\phi}$ an AME state if each marginal of $\lfloor\frac n2\rfloor$ parties is maximally mixed, where $\lfloor\cdot\rfloor$ is the floor function. For any AME state it isn't separable in any bipartition, and thus is a genuinely multipartite entangled (GME) state. Genuine entanglement, as a kind of special multipartite entanglement, is regarded as the most important resource, and has been used in various experiments \cite{gedense06,sdsexge09,8ghz11}. Hence, the characterization of multipartite entanglement has been widely investigated \cite{gmecri10,pmix10,ylge2019}. 

In recent years, AME states in homogeneous systems have aroused great interest due to the close connection between AME states and quantum error correction codes (QECCs) \cite{ameqecc04,ameqecc07,ameqecc2012,ameqecc2013,ameqecc2018,Huber_2018,ameqecc2019}. It is known that an AME state in the homogeneous system $\cH_d^{\ox n}$ is one-one corresponding to a QECC which encodes messages in an alphabet consisting of $d$ letters \cite{ameqecc04}. Moreover, AME states in homogeneous systems have been shown to be a resource for open-destination and parallel teleportation \cite{ameqecc2012}, for threshold quantum secret sharing schemes \cite{ameqecc2013} and can be used to design holographic quantum codes \cite{holographic2015}. In recent work the authors design a series of quantum circuits that generate absolute maximally entangled (AME) states to benchmark a quantum computer \cite{amecircuits2019}. 
Unfortunately, not all homogeneous systems contain AME states. For example, it has been shown that the AME states of seven qubits do not exist \cite{sevenqubitsame}. This key result completes the study on the existence of AME state of $n$ qubits. In an $n$-qubit system the AME states only exist when $n=2,3,5,6$ \cite{Huber_2018,amenote2018,ametable}. However, when the local dimensions are greater than $2$, it is still unknown whether there exist AME states in several homogeneous systems \cite{ametable}. 

There is no universal approach to construct AME states in homogeneous systems. Goyeneche \etal~ proposed a method to construct AME states in homogeneous systems by building the connection with orthogonal arrays (OAs) \cite{ameoa2014}. They proved that an irredundant orthogonal array (IrOA) is corresponding to an AME state in a homogeneous system. Then much effort has been done in this way to construct new AME states \cite{amemultiunitary,ameoa2018,kuniformoa2019,ameiroa2019}. Nevertheless, there is no efficient way yet to determine the existence of AME states if there exists no corresponding IrOA. 
By relaxing the restriction on AME states, we obtain a generalization of AME states called $k$-uniform states. A pure $n$-partite state is called $k$-uniform if every density matrix reduced to $k(\leq \lfloor\frac n2\rfloor)$ parties is maximally mixed. From the perspective of information theory, a $k$-uniform state in $\cH_d^{\ox n}$ has the property that all information about the system is lost after removal of $n-k$ or more qudits. It has been shown that $k$-uniform states can be constructed from graph states \cite{amegraph2013}, orthogonal arrays \cite{ameoa2014}, mutually orthogonal Latin squares and Latin cubes \cite{amemultiunitary}, and symmetric matrices \cite{kf2017}.   

In a practical quantum information processing, one may need to deal with a non-homogeneous quantum setup i.e., local dimensions of the quantum system are mixed. For example, the physical systems for encoding may have different numbers of energy levels. We call such a non-homogeneous system a heterogeneous one. Thus, it is natural to ask whether there exist AME states in heterogeneous systems. It is essential to study AME states in heterogeneous systems, since they are related to quantum error correcting codes over mixed alphabets \cite{mixedqecc13,ameshetero16}. Suppose $\ket{\psi}$ is a pure $N$-partite state in the heterogeneous system $\cH_{d_1}^{\ox n_1}\ox\cH_{d_2}^{\ox n_2}\ox\cdots\ox \cH_{d_l}^{\ox n_l}$, where $N=\sum_{i=1}^l n_i$. There are two definitions for $\ket{\psi}$ to be an AME state. One directly  follows from the definition of AME states in homogeneous systems \cite{ameshetero16}. It requires that every marginal of $\lfloor \frac{N}{2}\rfloor$ parties is maximally mixed. The other requires that every subsystem whose dimension is not larger than that of its complement must be maximally mixed \cite[Sec. 10]{Huber_2018}. In this work we select the former as the standard definition of AME states in heterogeneous systems.

Since the local dimensions are mixed, it significantly increases the difficulty of characterizing multipartite entanglement \cite{2mnlin06,mulitdiment2013}, and thus the experimental realization of GME states in heterogeneous systems \cite{233ges2016,twistedphotons2018}. Recently, a three-partite genuinely entangled state composed of one qubit and two qutrits has been experimentally realized \cite{233ges2016}. This is a remarkable step forward in generating GME states in a heterogeneous system. This experimental achievement motivates us to study the existence of AME states in heterogeneous systems. It is much more challenging to construct AME states in heterogeneous systems than in homogeneous systems because the heterogeneous systems are unruly and lack of efficient mathematical tools. There are only a few results on AME states in heterogeneous systems. Goyeneche \etal~ presented the concepts of mixed orthogonal arrays (MOAs) and irredundant mixed orthogonal arrays (IrMOAs) based on the concepts of OAs and IrOAs \cite{ameshetero16}. By the tool of MOA and IrMOA they constructed several concrete AME states in multipartite heterogeneous systems. Felix Huber \etal~ constructed a mixed-dimensional AME state in the system $2\times 3\times 3\times 3$ under the second definition above-mentioned \cite{Huber_2018}. Then in Ref. \cite[Sec. 5(b)]{amealgo2018} the authors numerically constructed four-partite mixed-dimensional AME state in the systems with maximal local dimension 4.

In this paper we focus on the construction of AME states in heterogeneous systems, especially tripartite heterogeneous AME states. Several properties of bipartite entanglement are well understood even if the bipartite system is heterogeneous. Nevertheless, many problems for tripartite systems become intractable. Although tripartite AME states in homogeneous systems are obvious, i.e., the GHZ state for any local dimension, it is not easy to construct tripartite AME states in heterogeneous systems, and not all tripartite heterogeneous systems contain AME states. The main aim of this work is to determine the existence of AME states in various tripartite heterogeneous systems, and construct several concrete AME states.

We first introduce the concept of irreducible AME states. The irreducible AME states cannot be written as the tensor product of two AME states. Thus irreducible AME states are the basic elements to generate AME states with high local dimensions. The existence of tripartite AME states in the heterogeneous systems possessing one qubit subsystem has been completely characterized by Lemma \ref{thm:2xmxn}. We present an alternative choice to prove Lemma \ref{thm:2xmxn} directly in Appendix \ref{sec:dirpf}. Since every AME state in systems $2\times kl\times (kl+l)$ is reducible, we also propose an approach to generate such tripartite AME states by Fig. \ref{fig:pro}. In Lemma \ref{le:trilmnch} we show the existence of AME states in several special heterogeneous systems. Then we mainly investigate the tripartite AME states in the heterogeneous systems whose local dimensions are $l,m,n$, with $3\leq l<m<n\leq m+l-1$. We propose a novel array called magic solution array (MSA). The restrictions for the elements of an MSA are similar to magic squares. In Theorem \ref{le:eslmml1} we show the construction of tripartite AME states in the aforesaid heterogeneous systems is closely related to the MSAs. Moreover, in Theorem \ref{le:irresuff} we present sufficient conditions for multipartite heterogeneous systems such that the AME states in them are irreducible. Finally we show our results are useful for the construction of $k$-uniform states possessing more parties, and quantum steering. In addition, we build the connection between AME states in heterogeneous systems and multi-isometry matrices.

The remainder of the paper is organized as follows. In Sec. \ref{sec:pre} we first clarify the notations in the whole paper, second formulate the mathematical definition of AME states in heterogeneous systems, and finally introduce the concept of irreducible AME states. In Sec. \ref{sec:3t} we determine the existence of tripartite AME states in various heterogeneous systems. In Sec. \ref{sec:irreames} we identify the AME states in which kinds of multipartite heterogeneous systems are irreducible. In Sec. \ref{sec:app} we show some applications of our results. Finally, the concluding remarks are given in Sec. \ref{sec:con}.



\section{Notations and Preliminaries}
\label{sec:pre}

Here we introduce the notations that will be used throughout the paper. Denote by $\lfloor\cdot\rfloor$ the interger part, by $\bbZ^{+}$ the set of positive integers, and by $\cU(d)$ the set of $d\times d$ unitary matrices. We denote a $d$-dimensional Hilbert space by $\cH_{d}$. An $n$-partite quantum system is represented by $\cH_{d_1}\ox\cH_{d_2}\ox\cdots\ox \cH_{d_n}$. For brevity we will refer to such a heterogeneous system as $d_1\times d_2\times \cdots\times d_n$.


The AME states are typically defined in homogeneous systems of $n$-qudit, i.e., $\cH_d^{\ox n}$. We first present the definition of AME states in homogeneous systems as basics. 
\begin{definition}
\label{def:ames}
Suppose $\ket{\Phi}$ is a pure $n$-partite state in the Hilbert space $\cH_d^{\ox n}$. Denote by $\cA_d(n)$ the set of AME states in $\cH_d^{\ox n}$. 

(i) $\ket{\Phi}$ is called $k$-uniform, if each $k$-partite marginal of $\proj{\Phi}$ is maximally mixed.

(ii) $\ket{\Phi}\in \cA_d(n)$ if it is $\lfloor\frac{n}{2}\rfloor$-uniform, i.e., each $\lfloor\frac{n}{2}\rfloor$-partite marginal of $\proj{\Phi}$ is maximally mixed.
\end{definition}


We next generalize the concept of AME states in homogeneous systems to that of AME states in heterogeneous systems. To be precise we shall present the mathematical definition as follows.

\begin{definition}
\label{def:amesmix}
Suppose $\ket{\Psi}$ is a pure $n$-partite state in the system $d_1\times d_2\times \cdots\times d_n$ where $d_i$'s are not equal. Denote by $\cA(d_1,d_2,\cdots,d_n)$ the set of AME states in this system.

(i) $\ket{\Psi}$ is called $k$-uniform, if each $k$-partite marginal of $\proj{\Psi}$ is maximally mixed, i.e., for any $\{j_1,\cdots,j_k\}\subset\{1,\cdots,n\}$,
\beq
\label{eq:defkunimix}
\tr_{\{j_1,\cdots,j_k\}^c}\proj{\Psi}=\frac{1}{d_{j_1}\times\cdots\times d_{j_k}} I_{d_{j_1}\times\cdots\times d_{j_k}}.
\eeq

(ii) $\ket{\Psi}\in \cA(d_1,d_2,\cdots,d_n)$, if it is $\lfloor\frac{n}{2}\rfloor$-uniform.
\end{definition}

Any $n$-partite AME state with even $n$ can be indeed regarded as a bipartite maximally entangled state with respect to any bipartition. However, it is impossible for the heterogeneous systems by Definition \ref{def:amesmix}. Thus there is no AME states in heterogeneous systems of even number of parties by Definition \ref{def:amesmix}. In the following when considering the existence of AME states in heterogeneous systems, we shall suppose the systems possessing odd number of parties.
In general the following lemma presents a necessary condition for the heterogeneous systems which contain $k$-uniform states.

\begin{lemma}
\label{le:dimensioncon}
A $k$-uniform state does not exist if the product of the size of $k$ local Hilbert spaces is larger than the dimension of the complementary system.
\end{lemma}
In Ref. \cite{Huber_2018}, the authors constructed a four-partite AME state in the system $2\times3\times3\times3$ based on a relaxed restriction. 



Suppose $\ket{\psi}_{A_1\cdots A_m}\in\cA(k_1,k_2,\cdots, k_m)$ and $\ket{\phi}_{B_1\cdots B_m}\in\cA(l_1,l_2,\cdots, l_m)$. One can verify that $\ket{\psi}\ox\ket{\phi}_{(A_1B_1)\cdots(A_mB_m)}$ is an AME state in $\cH_{k_1l_1}\ox \cdots\ox\cH_{k_ml_m}$, i.e., $\ket{\psi}\ox\ket{\phi}_{(A_1B_1)\cdots(A_mB_m)}\in\cA\big((k_1l_1),(k_2l_2),\cdots,(k_ml_m)\big)$. This property presents a method to construct an AME state in the system with larger local dimensions from two AME states in the systems with smaller local dimensions. 
Hence, it is essential to construct AME states $\ket{\ps}_{A_1\cdots A_m}$ and $\ket{\ph}_{B_1\cdots B_m}$ which have smaller local dimensions. For this purpose we establish the definition of irreducible AME states as follows. 

\begin{definition}
\label{df:irreducible}
Suppose $\ket{\ps}\in\cA(d_1,d_2,\cdots, d_n)$. Then $\ket{\psi}$ is called reducible AME state if there exists a local unitary $\ox_{i=1}^n U_i \in \cU(d_1)\ox \cU(d_2)\ox \cdots \ox \cU(d_n)$ such that $(\ox_{i=1}^n U_i)\ket{\ps}$ is the tensor product of two $n$-partite AME state with smaller local dimensions. Otherwise, we say that $\ket{\ps}$ is an irreducible AME state. 	
\end{definition}

For example, the three-qubit GHZ state $\ket{GHZ}={1\over\sqrt2}(\ket{000}+\ket{111})\in \cA_2(3)$ is an irreducible AME state,  while ${1\over2}(\ket{000}+\ket{111}+\ket{222}+\ket{333})\in\cA_4(3)$ is a reducible AME state as it is the tensor product of two GHZ states. 
The concept of irreducible AME states is very useful in experiments, since the tensor product which combines the corresponding subsystems of two states can be well realized by experiments.

In Ref. \cite{locallymaximally2019} the authors consider locally maximally entangled (LME) states, i.e., $1$-uniform states in kinds of heterogeneous systems. It follows that LME states are equivalent to AME states for tripartite heterogeneous systems. The existence of AME states in systems $2\times m\times (m+n)$ and $3\times m\times (m+n)$ has been characterized by \cite[Fig. 2]{locallymaximally2019}. In particular, the existence of AME states in systems $2\times m\times (m+n)$ can be characterized as follows.

\begin{lemma}\cite{locallymaximally2019}
\label{thm:2xmxn}
There exist tripartite AME states in systems $2\times m\times (m+n)$ if and only if $n=0$, or $m=kn, ~\forall k\geq 1$. 	
\end{lemma}
The authors derived Lemma \ref{thm:2xmxn} using the geometric invariant theory in Ref. \cite{locallymaximally2019}. We present an alternative proof in Appendix \ref{sec:dirpf} to show it directly.
Note that every tripartite AME states in the system $2\times kl\times (kl+l)$ is reducible. 
In Fig. \ref{fig:pro} we illustrate how to generate an AME state in the system $2\times kl\times (kl+l)$ with an irreducible AME state and a bipartite maximally entangled state. 
Nevertheless, we will show there exist both reducible AME states and irreducible AME states in systems $2\times m\times m$ in Sec. \ref{sec:irreames}.


\begin{figure}[ht]
\centering
\includegraphics[width=3in]{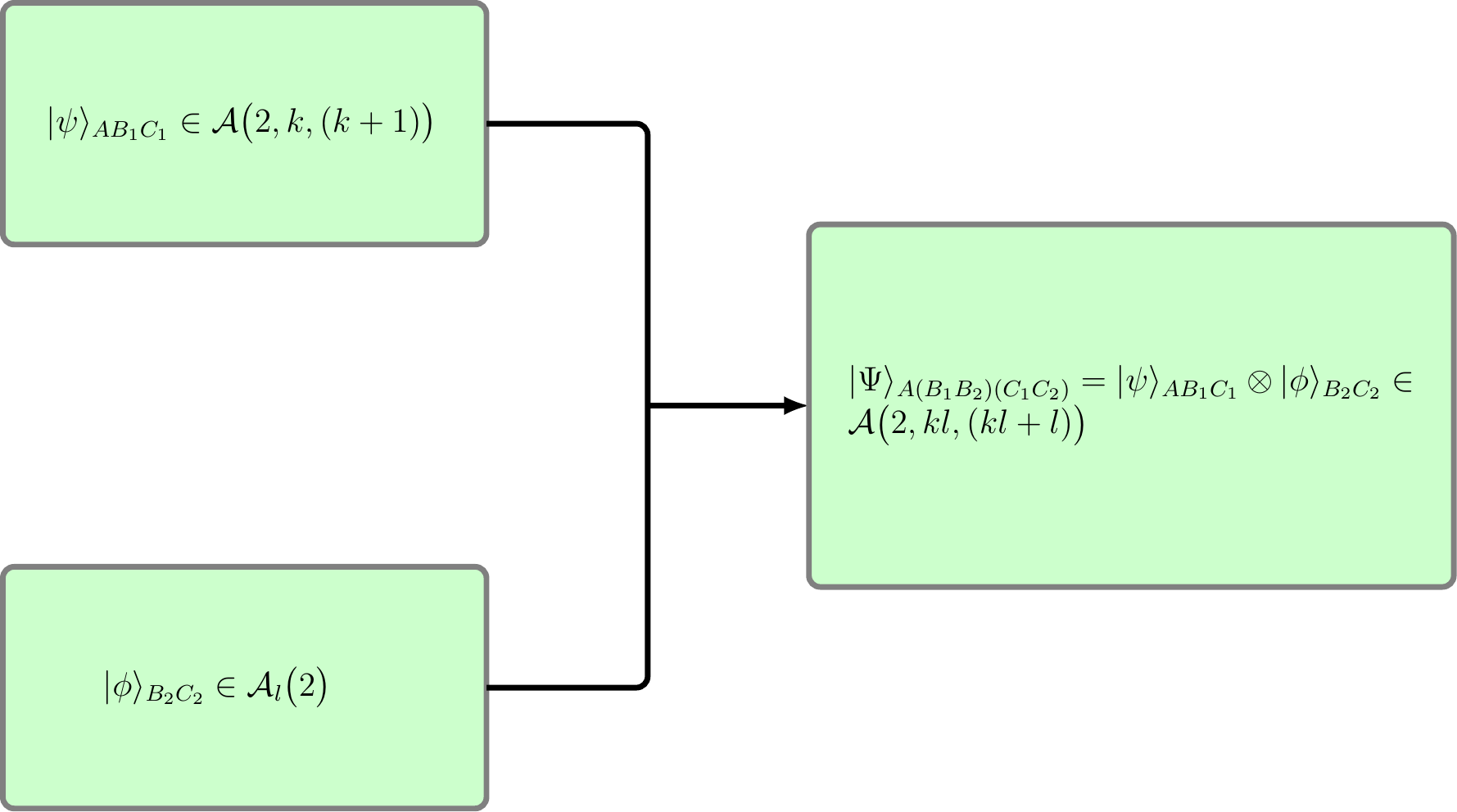}
\caption{\footnotesize{$\cA(2,k,(k+1))$ and $\cA_l(2)$ are the sets of AME states in systems $2\times k\times (k+1)$ and $l\times l$ respectively. The existence of $\ket{\psi}_{AB_1C_1}$ follows from Lemma \ref{thm:2xmxn}, and $\ket{\phi}_{B_2C_2}$ is indeed a bipartite maximally entangled state. Then $\ket{\Psi}_{A(B_1B_2)(C_1C_2)}$ is a tripartite AME state shared with $A,(B_1B_2), (C_1C_2)$. The tensor product combining the corresponding subsystems $B_1,B_2$, and $C_1,C_2$ has been widely used in experiments.}}
\label{fig:pro}
\end{figure}

\section{Tripartite AME states in heterogeneous systems}
\label{sec:3t}


In this section we investigate the tripartite AME states in heterogeneous systems further. The tripartite heterogeneous systems are the first non-trivial heterogeneous systems when considering AME states. It is reasonable to believe that the tripartite AME states in heterogeneous systems can be experimentally realized in the near future due to the experimental achievement of producing a tripartite genuinely entangled state in the heterogeneous system $2\times 3\times 3$ \cite{233ges2016}. The following lemma shows the existence of tripartite AME states in some special heterogeneous systems.

\begin{lemma}
\label{le:trilmnch}
(i) There exist AME states in systems $m\times m\times n, ~\forall n\leq m^2$.

(ii) Let $l<m$. There exist AME states in systems $l\times m\times n$ if $n=km\leq lm$.

(iii) Suppose both $\ket{\psi}_{ABC}$ and $\ket{\phi}_{ABC}$ belong to $\cA(d_A,d_B,d_C)$ where $d_A,d_B,d_C$ are the dimensions of the subsystems $A, B,C$ respectively. Define $\ket{\psi}\oplus_{AB}\ket{\phi}:=\frac{1}{\sqrt{2}}(\ket{00}_{AB}\ket{\psi}+\ket{11}_{AB}\ket{\phi})$. Then $\ket{\psi}\oplus_{AB}\ket{\phi}$ is a tripartite AME states in the system $(2d_A)\times(2d_B)\times(d_C)$.
\end{lemma}

\begin{proof}
(i) First from Lemma \ref{le:dimensioncon} the existence of AME states in systems $m\times m\times n$ requires $n\leq m^2$. Then one can construct an orthonormal basis $\{\ket{\Psi_j}\}_{j=0}^{m^2}$ of the system $m\times m$ which consists of $m^2$ maximally entangled states. Let 
\beq
\label{eq:mxmxn}
\ket{\Psi}=\frac{1}{\sqrt{n}}\sum_{j=0}^{n-1} \ket{\Psi_{j},j}.
\eeq 
By computing we have
\begin{eqnarray}
\notag
&&\rho_1=\frac{1}{n}\sum_{j=0}^{n-1}\tr_2\proj{\Psi_j}=\frac{1}{m}I_m,\\
\notag
&&\rho_2=\frac{1}{n}\sum_{j=0}^{n-1}\tr_1\proj{\Psi_j}=\frac{1}{m}I_m,\\
\notag
&&\rho_3=\frac{1}{n}\sum_{j=0}^{n-1}\proj{j}=\frac{1}{n}I_n,
\end{eqnarray}
where $\rho_i$ is the reduction of $i$th party.
By definition $\ket{\Psi}$ is a tripartite AME state in the system $m\times m\times n$.

(ii) Let $\o=e^{\frac{2\pi i}{l}}$, and $\ket{\O_l}=\frac{1}{\sqrt{l}}\sum_{x=0}^{l-1}\ket{xx}$. We construct the set of $l\times l$ diagonal unitary matrices
\beq
\label{eq:diagunitary}
D_t=\diag(1, \o^t,\cdots \o^{(l-1)t}),
\eeq
and the set of $m\times l$ matrices
\beq
\label{eq:permutation}
P_s=\sum_{j=0}^{l-1}\ketbra{j+s\mod m}{j}.
\eeq
Then we obtain the following states.
\beq
\label{eq:meslm}
\ket{\Psi_{t,s}}=(D_t\ox P_s)\ket{\O_l},
\eeq
where $t=0,\cdots,k-1$, and $s=0,\cdots,m-1$. One can verify $\ket{\Psi_{t,s}}, ~t=0,\cdots,k-1,~s=0,\cdots,m-1$ are orthogonal. Let 
\beq
\label{eq:lmkm}
\ket{\Psi}=\frac{1}{\sqrt{km}}\sum_{t=0}^{k-1}\sum_{s=0}^{m-1} \ket{\Psi_{t,s},mt+s}.
\eeq
One can verify every reduction of one party is maximally mixed. Therefore, $\ket{\Psi}$
is a tripartite AME state in the system $l\times m\times km$ with $0<k\leq l$.

(iii) Suppose $\rho_{ABC}$ is the density matrix of $\ket{\psi}\oplus_{AB}\ket{\phi}$. By computing it follows that
\begin{eqnarray}
\notag
&&\rho_A=\frac{1}{2}(\proj{0}_A\proj{\psi}_A+\proj{1}_A\proj{\psi}_A)=\frac{1}{2d_A}I_{2d_A},  \\
\notag
&&\rho_B=\frac{1}{2}(\proj{0}_B\proj{\psi}_B+\proj{1}_B\proj{\psi}_B)=\frac{1}{2d_B}I_{2d_B},  \\
\notag
&&\rho_C=\frac{1}{2}(\proj{\psi}_C+\proj{\psi}_C)=\frac{1}{d_C}I_{d_C}. 
\end{eqnarray} 
Therefore, $\ket{\psi}\oplus_{AB}\ket{\phi}$ is a tripartite AME states in the system $(2d_A)\times(2d_B)\times(d_C)$.

This completes the proof.
\end{proof}

One can prepare some states in \eqref{eq:mxmxn} using the current techniques in experiments. For example Alice, Bob and Charlie want to prepare the target state
\begin{eqnarray}
\ket{\Psi}={1\over\sqrt2}
\bigg(
\ket{\Psi_1}\ket{0}
+ 
\ket{\Psi_2}\ket{1}	
\bigg).
\end{eqnarray}
For this purpose Alice, Bob and Charlie may first prepare the following state
\begin{eqnarray}
\ket{\Phi}={1\over\sqrt2}
\bigg(
{1\over\sqrt2}(\ket{00}+\ket{11})\ket{0}
+ 
{1\over\sqrt2}(\ket{01}+\ket{10})\ket{1}	
\bigg).
\end{eqnarray}
Here we propose a protocol to prepare the state $\ket{\Phi}$ as follows.
Alice and Charlie may prepare the Bell state $\ket{\a}_{AB}={1\over\sqrt2}(\ket{00}+\ket{11})_{AB}$, and similarly Bob and Charlie may also prepare the Bell state $\ket{\a}_{BC}={1\over\sqrt2}(\ket{00}+\ket{11})_{BC}$. By setting $jk:=2j+k$ on system C, we obtain the tripartite state $\ket{\ps}:=\ket{\a}\otimes\ket{\b}={1\over2}(\ket{000}+\ket{011}+\ket{102}+\ket{113})$. Next Charlie measures system $C$ using the POVM $\{\proj{0}+\proj{3},\proj{1}+\proj{2}\}$. The result is that Charlie obtains the state ${1\over\sqrt2}(\ket{000}+\ket{113})$ or ${1\over\sqrt2}(\ket{011}+\ket{102})$. In either case Charlie may inform Alice and Bob of his measurement result, so that the latter may perform local unitary operations, so that the final tripartite state is the standard three-qubit GHZ state. This state is equivalent to the target state $\ket{\Phi}$ by local unitary operations. After preparing the state $\ket{\Phi}$ one may obtain the target state by $\ket{\Psi}=\ket{\Phi}\ox\ket{\psi}_{A'B'}$, where $\ket{\psi}_{A'B'}$ is a bipartite maximally entangled state in the composite system of Alice and Bob. Similarly, one may experimentally prepare more complex states in \eqref{eq:mxmxn}.

In the following we investigate the existence of tripartite AME states in general systems as $l\times m\times n$.
We first propose a kind of arrays whose restrictions are similar to magic squares. We call an $l\times m$ array as a magic solution array (MSA) if its elements $y_{k,j}, ~0\leq k\leq l-1, ~0\leq j\leq m-1$, constitute a nonnegative solution of the following system of linear equations:
\begin{numcases}{}
\label{eq:lmm+l-1rhoa-10}
\sum_{j=0}^{m-1} y_{k,j}=1, \quad \forall 0\leq k\leq l-1,  \\
\label{eq:lmm+l-1rhob-10}
\sum_{k=0}^{l-1} y_{k,j}=\frac{l}{m}, \quad \forall 0\leq j\leq m-1, \\
\label{eq:lmm+l-1rhoc-10}
\sum_{(k+j)\mod n=s} y_{k,j}=\frac{l}{n}, \quad \forall 0\leq s\leq n-1.
\end{numcases}

The following theorem shows that such MSAs are closely related to the construction of tripartite AME states in systems $l\times m\times n$.
\begin{theorem}
\label{le:eslmml1}
Suppose $3\leq l< m<n\leq m+l-1$. There exist AME states in the system $l\times m\times n$ if the magic solution array given by Eqs. \eqref{eq:lmm+l-1rhoa-10}-\eqref{eq:lmm+l-1rhoc-10} exists.
\end{theorem}

\begin{proof}
Suppose $\{y_{k,j}\}$ is a nonnegative solution of Eqs. \eqref{eq:lmm+l-1rhoa-10}-\eqref{eq:lmm+l-1rhoc-10}, where $0\leq k\leq l-1$, and $0\leq j\leq m-1$. Let $x_{k,j}=\sqrt{y_{k,j}}$. With these $x_{k,j}$'s we construct the following tripartite state $\ket{\psi}_{ABC}$ in the system $l\times m\times n$.
\begin{eqnarray}
\label{eq:explmml1}
&&\frac{1}{\sqrt{l}}\big[\ket{0}(\sum_{j=0}^{m-1}x_{0,j}\ket{j,j})+\ket{1}(\sum_{j=0}^{m-1}x_{1,j}\ket{j,j+1})
\\  \notag
&&+\cdots+\ket{n-m}(\sum_{j=0}^{m-1}x_{n-m,j}\ket{j,j+n-m})
\\ \notag
&&+\ket{n-m+1}(\sum_{j=0}^{m-1}x_{n-m+1,j}\ket{j,j+n-m+1\mod n})
\\ \notag
&&+\cdots+\ket{l-1}(\sum_{j=0}^{m-1}x_{l-1,j}\ket{j,j+l-1\mod n})\big].
\end{eqnarray}
Let $\rho_{ABC}=\proj{\psi}_{ABC}$. By computing we have the three single-party reductions as follows.
\begin{eqnarray}
\label{eq:lmm+l-1rhoa}
&& \rho_A=\frac{1}{l}(\sum_{k=0}^{l-1}\sum_{j=0}^{m-1}y_{k,j}\proj{k}),  \\
\label{eq:lmm+l-1rhob}
&& \rho_B=\frac{1}{l} (\sum_{j=0}^{m-1}\sum_{k=0}^{l-1}y_{k,j}\proj{j}),  \\
\label{eq:lmm+l-1rhoc}
&& \rho_C=\frac{1}{l} (\sum_{k=0}^{l-1}\sum_{j=0}^{m-1}y_{k,j}\proj{\overline{k+j}}),
\end{eqnarray}
where $\ket{\overline{k+j}}=\ket{k+j \mod n}$. Then one can verify that Eq. \eqref{eq:lmm+l-1rhoa-10} guarantees $\rho_A=\frac{1}{l}I_l$, Eq. \eqref{eq:lmm+l-1rhob-10} guarantees $\rho_B=\frac{1}{m}I_m$, and Eq. \eqref{eq:lmm+l-1rhoc-10} guarantees $\rho_C=\frac{1}{n} I_{n}$. Therefore, $\ket{\psi}_{ABC}$ expressed by Eq. \eqref{eq:explmml1} is a tripartite AME state in the system $l\times m\times n$. This completes the proof.
\end{proof}


By Theorem \ref{le:eslmml1} a given MSA is corresponding to an AME state in the system $l\times m\times n$ with $3\leq l< m<n\leq m+l-1$. For generic $l,m,n$ there is no rule to express the nonnegative solutions of the system of linear equations given by Eqs. \eqref{eq:lmm+l-1rhoa-10}-\eqref{eq:lmm+l-1rhoc-10}. Nevertheless, for specific $l,m,n$ one can numerically formulate the nonnegative solutions if the system of linear equations has nonnegative solutions.


Next, we formulate the expression of an AME state in the system $3\times 4\times 5$ as an example using Theorem \ref{le:eslmml1}. Since all AME states in the system $3\times 4\times 5$ are irreducible, this system becomes the first interesting case when the dimension of every subsystem is greater than $2$.

\begin{example}
\label{le:existence345}
We propose a concrete AME state in the system $3\times 4\times 5$. Let 
\beq
\label{eq:psi345}
\bal
\ket{\psi}_{ABC}&=\frac{1}{\sqrt{3}}\ket{0}(\sum_{j=0}^3x_{0j}\ket{j,j})+\ket{1}(\sum_{j=0}^3x_{1j}\ket{j,j+1})\\
&+\ket{2}(\sum_{j=0}^3x_{2j}\ket{j,j+2\mod 5}).
\eal
\eeq
Let $\rho_{ABC}=\proj{\psi}_{ABC}$. By computing we have the three single-party marginals as follows.
\begin{eqnarray}
\label{eq:345rhoA}
&& \rho_A=\frac{1}{3}\sum_{k=0}^2(\sum_{j=0}^3\abs{x_{kj}}^2)\proj{k},  \\
\label{eq:345rhoB}
&& \rho_B=\frac{1}{3}\big((\sum_{k=0}^2\abs{x_{k0}}^2)\proj{0}+(\sum_{k=0}^2\abs{x_{k1}}^2)\proj{1} \\
\notag
&& +(\sum_{k=0}^2\abs{x_{k2}}^2)\proj{2}+(\sum_{k=0}^2\abs{x_{k3}}^2)\proj{3} \big)   \\
\label{eq:345rhoC}
&& \rho_C= \frac{1}{3} \big((\sum_{\overline{k+j}=0} \abs{x_{kj}}^2) \proj{0}\\  \notag
&&+(\sum_{\overline{k+j}=1} \abs{x_{kj}}^2) \proj{1}+(\sum_{\overline{k+j}=2} \abs{x_{kj}}^2) \proj{2}\\ \notag
&&+(\sum_{\overline{k+j}=3} \abs{x_{kj}}^2) \proj{3}+(\sum_{\overline{k+j}=4} \abs{x_{kj}}^2) \proj{4}\big),
\end{eqnarray}
where $\overline{k+j}=(k+j\mod 5)$.
Let $y_{k,j}=\abs{x_{kj}}^2$, and define an array $Y:=[y_{k,j}],~0\leq k\leq 2, ~0\leq j\leq 3$. So $\rho_A=\frac{1}{3}I_3,\rho_B=\frac{1}{4}I_4,\rho_C=\frac{1}{5}I_5$ are equivalent to the following three equations respectively.
\begin{eqnarray}
\label{eq:345rhoA-1}
&& \sum_{j=0}^{3} y_{k,j}=1, \quad \forall 0\leq k\leq 2,  \\
\label{eq:345rhoB-1}
&& \sum_{k=0}^{2} y_{k,j}=\frac{3}{4}, \quad \forall 0\leq j\leq 3, \\
\label{eq:345rhoC-1}
&& \sum_{\overline{k+j}=n} y_{k,j}=\frac{3}{5}, \quad \forall 0\leq n\leq 4.
\end{eqnarray}
One can verify the following array
\beq
\label{eq:345solutionY}
Y=
\bma
\frac{12}{40} & \frac{24}{40} & \frac{4}{40} & 0 \\
0 & \frac{2}{40} & \frac{20}{40} & \frac{18}{40} \\
\frac{18}{40} & \frac{4}{40} & \frac{6}{40} & \frac{12}{40}
\ema
\eeq
is a MSA corresponding to the system $3\times 4\times 5$. Therefore, $\ket{\psi}_{ABC}$ expressed by Eq. \eqref{eq:psi345} whose coefficients given by the magic solution array $Y$ in Eq. \eqref{eq:345solutionY} is an AME state in the system $3\times 4\times 5$.
\qed
\end{example}







\section{Irreducible AME states in heterogeneous systems}
\label{sec:irreames}

The irreducible AME states are essential blocks for constructing AME states in heterogeneous systems. In this section we investigate the AME states in which kinds of heterogeneous systems are irreducible.
In Theorem \ref{le:irresuff} (i) we keep focusing on the tripartite AME states, and in Theorem \ref{le:irresuff} (ii) and (iii) we study the multipartite heterogeneous systems.

\begin{theorem}
\label{le:irresuff}
(i) Suppose $p$ is prime and $m,n$ are coprime. If $\ket{\psi}\in\cA(p,m,n)$ then $\ket{\psi}$ is an irreducible AME state.


(ii) Suppose $\ket{\psi}$ is an AME state in the system $p\times q\times d_1\times \cdots\times d_{2n-1}$, where $p$ and $q$ are prime. Then $\ket{\psi}$ is irreducible if there exists $d_i\neq pq$. Furthermore for $n=1$, $\ket{\psi}$ is irreducible if and only if there exists $d_1< pq$.

(iii) Suppose $\ket{\psi}$ is an AME state in the system $d_1\times\cdots\times d_{2n+1}$. If $\ket{\psi}$ is reducible, then there are at most two primes in $\{d_1,\cdots, d_{2n+1}\}$. Further, suppose $\ket{\psi}$ is locally unitarily equivalent to $\ket{\phi_1}\ox\ket{\phi_2}$, where $\ket{\phi_1}$ is an AME state in the system $p_1\times\cdots\times p_{2n+1}$, $\ket{\phi_2}$ is an AME state in the system $q_1\times\cdots\times q_{2n+1}$, and $d_i=p_iq_i$. Then,

(iii.a) if there is no prime in $\{d_1,\cdots, d_{2n+1}\}$, the number of $1$ in $\{p_1,\cdots,p_{2n+1}\}$ is at most one, so is $\{q_1,\cdots,q_{2n+1}\}$;

(iii.b) if only $d_1$ is prime up to a permutation of subsystems, then $p_1=1, q_1=d_1$, and $p_i,q_j>1, ~\forall i,j>1$;

(iii.c) if only $d_1$ and $d_2$ are prime up to a permutation of subsystems, then $p_1=1, q_1=d_1$, $p_2=d_2,q_2=1$, and $p_i,q_j>1, ~\forall i,j>2$.
\end{theorem}

\begin{proof}
(i) We prove it by contradiction. Suppose $\ket{\psi}$ is a reducible AME state. By Definition \ref{df:irreducible}, there exists a local unitary $U\ox V\ox W$ such that $(U\ox V\ox W)\ket{\psi}=\ket{\phi_1}\ox \ket{\phi_2}$, where $\ket{\phi_1}$ is an AME state in the system $1\times m_1\times n_1$ and $\ket{\phi_2}$ is an AME state in $p\times m_2\times n_2$. It implies that $m_1=n_1$ which contradicts with $m$ and $n$ are coprime. Therefore, $\ket{\psi}$ is irreducible.


(ii) We prove it by contradiction. Suppose $\ket{\psi}$ is a reducible AME states. By definition $\ket{\psi}$ is locally unitarily equivalent to $\ket{\phi_1}\ox\ket{\phi_2}$, where $\ket{\phi_1}$ is an AME state in the system $p\times 1\times k_1\times\cdots\times k_{2n-1}$, $\ket{\phi_2}$ is an AME state in $1\times q\times l_1\times\cdots\times l_{2n-1}$, and $d_i=k_il_i$. So $\ket{\phi_1}$ and $\ket{\phi_2}$ can be taken as $2n$-partite AME state. Since there is no AME state in heterogeneous systems of even number of parties, it follows that $k_i=p, l_j=q, \forall i,j$. It is equivalent to that $d_i=pq,\forall i$. So we obtain the contradiction, and thus the assertion holds. Furthermore for $n=1$, we have $\ket{\psi}$ is a tripartite reducible AME state if and only if $m=pq$.

(iii) We prove it by contradiction. Suppose $d_1,d_2,d_3$ are three primes, and $\ket{\psi}$ is a reducible AME state in the system $d_1\times\cdots\times d_{2n+1}$. By definition $\ket{\psi}$ is locally unitarily equivalent to $\ket{\phi_1}\ox\ket{\phi_2}$, where $\ket{\phi_1}$ is an AME state in the system $p_1\times\cdots\times p_{2n+1}$, $\ket{\phi_2}$ is an AME state in the system $q_1\times\cdots\times q_{2n+1}$, and $d_i=p_iq_i$. If $p_1=p_2=1$, then $\ket{\phi_1}$ can be taken as a $(2n-1)$-partite state, and thus it isn't a $(2n+1)$-partite AME state. One can similarly exclude that $q_1=q_2=1$. Therefore, there is at most one $1$ in the set $\{p_1,p_2,p_3\}$, the same for the set $\{q_1,q_2,q_3\}$. It contradicts with the three $d_1,d_2,d_3$ are prime. So the assertion (iii) holds. When there are at most two primes in $\{d_1,\cdots, d_{2n+1}\}$, one can verify the three specific cases (iii.a)-(iii.c) with the same idea.

This completes the proof.
\end{proof}

It follows from Theorem \ref{le:irresuff} (iii) that if there are at least three primes among $\{d_1,\cdots,d_{2n+1}\}$, then every AME state in the system ${d_1}\times\cdots\times {d_{2n+1}}$ is irreducible.
Here we want to emphasize that if there is a reducible AME state in the system $d_1\times\cdots\times d_{n}$, it doesn't imply that every AME state in the system $d_1\times\cdots\times d_{n}$ is reducible, i.e., there could be irreducible AME states in the system $d_1\times\cdots\times d_{n}$. The following example supports our claim.
First, it is known that there exist reducible AME states in the system $2\times 4\times 4$, e.g., $(\ket{000}+\ket{111})_{AB_1C_1}\ox (\ket{00}+\ket{11})_{B_2C_2}$. Second, We show that $\ket{\psi}_{ABC}=\ket{0,x}+\ket{1,y}$, where $\ket{x}=\frac{1}{2}(\ket{00}+\ket{11}+\ket{22}+\ket{33})$ and $\ket{y}=\frac{1}{2}(\ket{01}+\ket{12}+\ket{23}+\ket{30})$ is an irreducible AME state in the system $2\times 4\times 4$. One can show that the range of $\ket{x}$ and $\ket{y}$ has no bipartite state of Schmidt rank two in $\cH_{BC}$. By computing one can verify $\ket{\psi}$ is an AME state. We next show it is irreducible. Assume that $\ket{\psi}$ is reducible. It follows that $\ket{\psi}=(\ket{0,a}+\ket{1,b})_{AB_1C_1}\ox \ket{c}_{B_2C_2}$ such that system $B=B_1B_2$ and system $C=C_1C_2$. It implies that $(\ket{0,a}+\ket{1,b})$ is three-qubit state and $\ket{c}$ is a two-qubit state. It is known that the span of $\ket{a}$ and $\ket{b}$ has a product vector. Hence, the span of $\ket{a,c}$ and $\ket{b,c}$ has a bipartite state of Schmidt rank two in $\cH_{BC}$. Since $\ket{a,c}=\ket{x}, ~\ket{b,c}=\ket{y}$, we obtain the contradiction. Therefore, $\ket{\psi}$ is an irreducible AME state in the system $2\times 4\times 4$. Hence, there are both reducible and irreducible AME states in the system $2\times 4\times 4$.


\section{Applications}
\label{sec:app}

In this section we indicate some applications of our results in previous sections. In Sec. \ref{subsec:conskuniform} we propose some methods to construct $k$-uniform states of more parties from two AME states. In. Sec. \ref{sec:multiu} we build the connection between AME states in heterogeneous systems and multi-isometry matrices. In Sec. \ref{subsec:steering} we introduce the application of AME states in quantum steering.

\subsection{Construction of $k$-uniform states in heterogeneous systems}
\label{subsec:conskuniform}

Here we present some methods to construct $k$-uniform states using two existing AME states. The methods in Lemma \ref{le:2nto2n+1} (ii) and (iii) are also very useful to construct other kinds of GME states \cite{ylge2019}.

\begin{lemma}
\label{le:2nto2n+1}
(i) If $\ket{\psi}$ is a $(2n)$-partite AME state in the composite system of $(AB),C_1,C_2,\cdots C_{2n-1}$, it is also a $(2n+1)$-partite AME state in the system of $A,B,C_1,C_2,\cdots C_{2n-1}$.

(ii) Suppose $\ket{\psi}$ is a $(2n)$-partite AME state in the system of $A, C_{1,1},\cdots,C_{1,2n-1}$, and $\ket{\phi}$ is a $(2n)$-partite AME state in the system of $B, C_{2,1},\cdots,C_{2,2n-1}$. Then $\ket{\psi}\ox\ket{\phi}$ is a $(2n+1)$-partite AME state in the system of $A,B,C_1,\cdots, C_{2n-1}$, where $C_j=(C_{1,j}C_{2,j})$.

(iii) Suppose $\ket{\psi}$ is a $(2n+1)$-partite AME state in the system of $A, C_{1,1},\cdots,C_{1,2n}$, and $\ket{\phi}$ is a $(2n+1)$-partite AME state in the system of $B, C_{2,1},\cdots,C_{2,2n}$. Then $\ket{\psi}\ox\ket{\phi}$ is an $n$-uniform state in the system of $A,B,C_1,\cdots, C_{2n-1}$, where $C_j=(C_{1,j}C_{2,j})$.
\end{lemma}

\begin{proof}
(i) Suppose $\rho=\proj{\psi}$. By computing,
\begin{eqnarray}
\label{eq:rhoc-1}
\notag
&& \rho_{C_{j_1}\cdots C_{j_n}}\propto I,   \\
\notag
\label{eq:rhoabc}
&& \rho_{ABC_{j_1}\cdots C_{j_{n-2}}}\propto I,   \\
\notag
\label{eq:rhoac}
&& \rho_{AC_{j_1}\cdots C_{j_{n-1}}}=\tr_B\rho_{ABC_{j_1}\cdots C_{j_{n-1}}}\propto I,  \\
\notag
\label{eq:rhobc}
&& \rho_{BC_{j_1}\cdots C_{j_{n-1}}}=\tr_A\rho_{ABC_{j_1}\cdots C_{j_{n-1}}}\propto I.
\end{eqnarray}
Therefore, $\ket{\psi}$ is also a $(2n+1)$-partite AME state in the system of $A,B,C_1,C_2,\cdots C_{2n-1}$.

(ii) Suppose $\sigma=\proj{\psi}$, $\gamma=\proj{\phi}$, and $\rho$ is the density matrix of $\ket{\psi}\ox\ket{\phi}$. By computing we have
\begin{eqnarray}
\label{eq:rhoc}
\notag
&& \rho_{C_{j_1}\cdots C_{j_n}}=\sigma_{C_{1,j_1}\cdots C_{1,j_n}}\ox\gamma_{C_{2,j_1}\cdots C_{2,j_n}}\propto I,   \\
\notag
\label{eq:rhoac}
&& \rho_{AC_{j_1}\cdots C_{j_{n-1}}}=\sigma_{AC_{1,j_1}\cdots C_{1,j_{n-1}}}\ox\gamma_{C_{2,j_1}\cdots C_{2,j_{n-1}}}\propto I,  \\
\notag
\label{eq:rhobc}
&& \rho_{BC_{j_1}\cdots C_{j_{n-1}}}=\sigma_{C_{1,j_1}\cdots C_{1,j_{n-1}}}\ox\gamma_{BC_{2,j_1}\cdots C_{2,j_{n-1}}}\propto I,  \\
\notag
\label{eq:rhoabc}
&& \rho_{ABC_{j_1}\cdots C_{j_{n-2}}}=\sigma_{AC_{1,j_1}\cdots C_{1,j_{n-2}}}\ox\gamma_{BC_{2,j_1}\cdots C_{2,j_{n-2}}}\propto I.
\end{eqnarray}
Therefore, $\ket{\psi}\ox\ket{\phi}$ is a $(2n+1)$-partite AME state in the system of  $A,B,C_1,\cdots, C_{2n-1}$, where $C_j=(C_{1,j}C_{2,j})$.

(iii) The proof is similar to (ii).

This completes the proof.
\end{proof}

\subsection{AME states in heterogeneous systems and multi-isometry matrices}
\label{sec:multiu}

In Ref. \cite{amemultiunitary} the authors introduced the concept of {\em multiunitary matrices}, and build the connection between it and AME states in homogeneous systems. In this subsection we first recall the multiunitarity property and then generalize it to the concept of multi-isometry matrices for AME states in heterogeneous systems.

The square matrix $A$ of order $d^k$ acting on a composed Hilbert space $\cH_d^{\ox k}$, and represented by
\beq
\label{eq:multiu}
(A)_{\m_1,\cdots,\m_k\atop\n_1,\cdots,\n_k}=\bra{\m_1,\cdots,\m_k}A\ket{\n_1,\cdots,\n_k}=a_{\m_1,\cdots,\m_k\atop\n_1,\cdots,\n_k}
\eeq
is called $k$-unitary if it is unitary for all possible ${2k\choose k}$ reordering of its indices, corresponding to all possible choices of $k$ indices out of $2k$ \cite{amemultiunitary}. Here $\m_i,\n_j=0,\cdots,d-1$, and each forms an orthonormal basis of $\cH_d$. Then one can construct the following unnormalized pure state in the Hilbert space $\cH_d^{\ox 2k}$
\beq
\label{eq:multiuame}
\ket{\phi}=\sum_{\m_1,\cdots,\m_k,\atop\n_1,\cdots,\n_k=0}^{d-1}a_{\m_1,\cdots,\m_k\atop\n_1,\cdots,\n_k}\ket{\m_1,\cdots,\m_k,\n_1,\cdots,\n_k}.
\eeq
It follows from the multiunitarity property that $\ket{\phi}$ in Eq. \eqref{eq:multiuame} is an unnormalized AME state in $\cH_d^{\ox 2k}$. In special, if $A$ is a $1$-unitary matrix, it is a standard unitary matrix. So the first interesting case is $2$-unitary matrices. Goyeneche \etal constructed a concrete $2$-unitary matrix using the AME state in $\cA(4,3)$, and derived several results on $2$-unitary matrices \cite{amemultiunitary}.

As the above discussion, the multiunitary matrices are closely related to AME states in homogeneous systems of even parties. In a direct generation, we introduce the multi-isometry matrices which are connected to the AME states shared with odd parties. A matrix $M$ is called isometry if $M^\dg M=I$. Then we can similarly define the {\em multi-isometry} matrices.

The matrix $A$ represented by
\beq
\label{eq:multiidef}
(A)_{\m_1,\cdots,\m_{k+1}\atop\n_1,\cdots,\n_{k}}=\bra{\m_1,\cdots,\m_{k+1}}A\ket{\n_1,\cdots,\n_{k}}=a_{\m_1,\cdots,\m_{k+1}\atop\n_1,\cdots,\n_{k}}
\eeq
is called $k$-isometry if $\tr(A^\dg A)$ is constant and $A^\dg A\propto I$ for all possible ${2k+1\choose k}$ reordering of its indices, corresponding to all possible choices of $k$ indices out of $(2k+1)$. Here $\m_i=0,\cdots,d_i-1,$ which is an orthonormal basis of $\cH_{d_i}$, and $\n_j=0,\cdots,l_j-1,$ which is an orthonormal basis of $\cH_{l_j}$. As an extension we present the following lemma.
\begin{lemma}
\label{le:amesmultiiso}
Suppose $A$ is a $k$-isometry matrix whose elements $a_{\m_1,\cdots,\m_{k+1}\atop\n_1,\cdots,\n_k}$ are given by Eq. \eqref{eq:multiidef}. Then the following $(2k+1)$-partite state 
\beq
\label{eq:amemultiiso1}
\ket{\phi}=\sum_{\m_1,\cdots,\m_{k+1},\atop\n_1,\cdots,\n_k=0}^{d-1}a_{\m_1,\cdots,\m_{k+1}\atop\n_1,\cdots,\n_k}\ket{\m_1,\cdots,\m_{k+1},\n_1,\cdots,\n_k}
\eeq
is an unnormalized AME state in the system $d_1\times\cdots\times d_{k+1}\times l_1\times\cdots\times l_k$.
\end{lemma}
This lemma follows from the fact that each $k$-partite marginal of $\proj{\phi}$ is equal to $A^\dg A$, where $A$ is represented in a product basis corresponding to a choice of $k$ indices out of $(2k+1)$. Since $A$ is $k$-isometry, each $k$-partite marginal of $\proj{\phi}$ is proportional to the identity. Thus $\ket{\phi}$ in Eq. \eqref{eq:amemultiiso1} is an unnormalized AME states.

It is different from the multiunitarity property that $1$-isometry matrices are not equivalent to the standard isometry matrices. In other words, there exist standard isometry matrices which are not $1$-isometry. The $1$-isometry property requires that the matrices constructed by $3$ different ways of choosing a bipartition of the indices are all proportional to the identity. Hence, it's interesting to study the relation between $1$-isometry matrices and tripartite AME states. In order to better understand the multi-isometry property, we construct a concrete $1$-isometry matrix using the tripartite AME state in the system $3\times 4\times 5$ which is in Example \ref{le:existence345}. By multiplying a coefficient we write it as
\beq
\label{eq:psi345unn}
\bal
\ket{\psi}_{ABC}&=\sum_{\m_1=0}^2\sum_{\m_2=0}^3\sum_{\n_1=0}^4 a_{\m_1,\m_2\atop \n_1}\ket{\m_1,\m_2,\n_1}\\
&=\sqrt{\frac{1}{2}}\ket{000}+\ket{011}+\sqrt{\frac{1}{6}}\ket{022}\\
&+\sqrt{\frac{1}{12}}\ket{112}+\sqrt{\frac{5}{6}}\ket{123}+\sqrt{\frac{3}{4}}\ket{134}\\
&+\sqrt{\frac{3}{4}}\ket{202}+\sqrt{\frac{1}{6}}\ket{213}+\sqrt{\frac{1}{4}}\ket{224}+\sqrt{\frac{1}{2}}\ket{230}.
\eal
\eeq
By computing we have 
\beq
\label{eq:trimultii345}
\bal
A^0&=\sum_{\m_1=0}^2\sum_{\m_2=0}^3\sum_{\n_1=0}^4 a_{\m_1,\m_2\atop \n_1}\ketbra{\m_1,\m_2}{\n_1},~ (A^0)^\dg(A^0)=I_5,\\
A^1&=\sum_{\m_1=0}^2\sum_{\m_2=0}^3\sum_{\n_1=0}^4 a_{\m_2,\n_1\atop \m_1}\ketbra{\m_2,\n_1}{\m_1}, ~(A^1)^\dg(A^1)=\frac{5}{3}I_3,\\
A^2&=\sum_{\m_1=0}^2\sum_{\m_2=0}^3\sum_{\n_1=0}^4 a_{\n_1,\m_1\atop \m_2}\ketbra{\n_1,\m_1}{\m_2}, ~(A^2)^\dg(A^2)=\frac{5}{4}I_4.\\
\eal
\eeq
Thus, by definition each $A^0,A^1,A^2$ in Eq. \eqref{eq:trimultii345} is $1$-isometry.

\subsection{Quantum steering for heterogeneous systems}
\label{subsec:steering}

Steering has been found useful in a number of applications such as subchannel discrimination and one-sided deviceindependent quantum key distribution. Thus, detection and characterization of steering have recently attracted increasing attention. In Ref. \cite{steerct16} the authors propose a general framework for constructing universal steering criteria that are applicable to arbitrary measurement settings of the steering party.
Here we introduce the quantum steering as an application of the tripartite AME states in \eqref{eq:mxmxn}. Suppose it is controlled by the system Alice, Bob and Charlie. If
$n=m^2$, then Alice and Bob are in the maximally mixed state $\r_{AB}={1\over n^2}I_{n^2}$. This is a separable state, and also a classical-classical state \cite{ccm11}. Using the projective POVM $\{\proj{j},j=1,...,n^2\}$, Charlie can steer the state $\r_{AB}$ into the maximally entangled state $\ket{\Ps_j}$ with probability $1/n^2$. Since any two maximally entangled states are LU equivalent, Alice and Bob can convert $\ket{\Ps_j}$ into the standard maximally entangled states ${1\over\sqrt n}\sum^n_{j=1}\ket{jj}$. One can show that the same argument works when $n<m^2$, though $\r_{AB}$ may be not separable.

Similarly, Alice (or Bob) may perform the projective POVM $\{\proj{j},j=1,...,n\}$ on system $A$ (or $B$), so as to steer the state $\r_{BC}$ of Bob and Charlie (or $\r_{AC}$ of Alice and Charlie) into the standard maximally entangled state up to LU equivalence. Since $\r_{BC}$ and $\r_{AC}$ are both rank-$n$ mixed entangled states, the steering is a kind of as system-assisted and one-copy entanglement distillation with probability one \cite{assisteddis13}. 




\section{Concluding remarks}
\label{sec:con}


In this paper we mainly investigated AME states in tripartite heterogeneous systems.
First we introduced the concept of irreducible AME states as the essential elements to construct AME states with high local dimensions.
Then we derived the expressions of AME states in a subset of general tripartite heterogeneous systems based on the expressions of AME states in systems $2\times m\times (m+n)$. We showed that the existence of AME states in systems $l\times m\times n$ with $3\leq l<m<n\leq m+l-1$ is related to the corresponding MSAs.
Moreover, we studied the irreducible AME in multipartite heterogeneous systems deeply. We presented sufficient conditions for multipartite heterogeneous systems such that the AME states in them are irreducible.
Finally we have shown our results are useful for the construction of $k$-uniform states of more parties and can be applied to quantum steering. We additionally revealed the connection between heterogeneous AME states and multi-isometry matrices.

There are still some open problems for AME states in tripartite heterogeneous systems. Although for a specific system $l\times m\times n$ the corresponding MSA can be found numerically if it exists, there is no criteria for the parameters $l,m,n$ such that the systems $l\times m\times n$ contain the corresponding MSAs. 
If there is no corresponding MSA, we need to find other efficient ways to construct AME states in systems $l\times m\times n$. It is also an interesting problem to determine whether a heterogeneous system includes both reducible and irreducible AME states. Lastly one may ask whether there exist AME states in multipartite heterogeneous systems. A possible direction for this problem is to find the multi-isometry matrices described in Sec. \ref{sec:multiu}.

\section*{Acknowledgments}

We appreciate Felix Huber for his valuable suggestions. LC and YS were supported by the NNSF of China (Grant No. 11871089), and the Fundamental Research Funds for the Central
Universities (Grant Nos. KG12080401, and ZG216S1902).

\appendix

\section{Direct Proof of Lemma \ref{thm:2xmxn}}
\label{sec:dirpf}

\begin{proof}
First, if $n=0$, the existence of AME states in systems $2\times m\times m$ follows directly from Lemma \ref{le:trilmnch} (i).

Second, we investigate the case when $n\geq 1$.
Suppose $\ket{\psi}_{ABC}$ is an arbitrary tripartite AME state in the system $2\times m\times (m+n)$. Up to a local unitary operation, we can assume that 
\begin{eqnarray}
\label{eq:psi=2xmxn}	
\ket{\psi}_{ABC}=\frac{1}{\sqrt{2}}[\ket{0}(\sum_{j=0}^{m-1}x_j\ket{j,j})+\ket{1}(\sum_{j=0}^{m-1}y_j \ket{\a_j,\b_j})] 
\end{eqnarray}
where $x_j,y_j\ge0$ with $\sum_jx_j^2=\sum_jy_j^2=1$, and $\{\ket{\a_j}\}$ and $\{\ket{\b_j}\}$ are two sets of orthonormal vectors in Hilbert spaces $\cH_m$ and $\cH_{(m+n)}$, respectively. Let $\rho=\proj{\psi}$ be the density matrix of $\ket{\psi}_{ABC}$. By computing its single-party reductions we have
\begin{eqnarray}
\label{eq:2mm+nrhoa}
&& \rho_A=\frac 12 I_2,  \\
\label{eq:2mm+nrhob}
&& \rho_B=\frac{1}{2} (\sum_{j=0}^{m-1}x_j^2\proj{j})+\frac{1}{2} (\sum_{j=0}^{m-1} y_j^2 \proj{\a_j}),  \\
\label{eq:2mm+nrhoc}
&& \rho_C=\frac{1}{2} (\sum_{j=0}^{m-1}x_j^2\proj{j})+\frac{1}{2} (\sum_{j=0}^{m-1} y_j^2 \proj{\b_j}).
\end{eqnarray}
By Definition \ref{def:amesmix} it requires $\rho_B=\frac 1m I_m$ and $\rho_C=\frac{1}{m+n} I_{m+n}$.
Since $m\ge n\ge1$, it follows from Eq. \eqref{eq:2mm+nrhoc} that at least $n$ elements of $\{x_j^2\}_{j=0}^{m-1}$ are equal to $\frac{2}{m+n}$. Up to a local unitary transformation on $\cH^B\otimes\cH^C$, we can assume that $x_0^2=x_1^2=\cdots=x_{n-1}^2=\frac{2}{m+n}$. Using Eqs. \eqref{eq:2mm+nrhob}-\eqref{eq:2mm+nrhoc}, we obtain the set $\big\{y_0^2,\cdots, y_{m-1}^2\big\}$ and the following two sets are the same set. 
\beq
\label{eq:contraction-1}
\bal
&\big\{\underbrace{\frac{2n}{m(m+n)},\cdots, \frac{2n}{m(m+n)}}_n, \\
&\frac{2}{m}-x_n^2,\frac{2}{m}-x_{n+1}^2,\cdots, \frac{2}{m}-x_{m-1}^2\big\},
\eal
\eeq
\beq
\label{eq:contraction-2}
\bal
&\big\{\frac{2}{m+n}-x_n^2,\frac{2}{m+n}-x_{n+1}^2,\cdots, \frac{2}{m+n}-x_{m-1}^2,\\
&\underbrace{\frac{2}{m+n},\cdots,\frac{2}{m+n}}_n\big\}.
\eal
\eeq

We first assume $d=(m\mod n)$, and $d>0$. Without loss of generality, we may assume 
$$\frac{2}{m+n}-x_n^2=\cdots=\frac{2}{m+n}-x_{2n-1}^2=\frac{2n}{m(m+n)}.$$
Thus $x_n^2=\cdots=x_{2n-1}^2=\frac{2m-2n}{m(m+n)}$. If $\frac{2}{m}-x_n^2=\cdots=\frac{2}{m}-x_{2n-1}^2=\frac{2}{m+n}$, it implies $\frac{4n}{m(m+n)}=\frac{2}{m+n}$ which contradicts with $d>0$. Hence, we may further assume 
\begin{eqnarray}
\notag
&&\frac{2}{m}-x_n^2=\cdots=\frac{2}{m}-x_{2n-1}^2\\
\notag
&&=\frac{2}{m+n}-x_{2n}^2=\cdots\frac{2}{m+n}-x_{3n-1}^2.
\end{eqnarray}
It implies the following recursive relation $x_{(j-1)n}^2-x_{jn}^2=\frac{2n}{m(m+n)}$. Since the two sets given by \eqref{eq:contraction-1} and \eqref{eq:contraction-2} are equal, by computing it requires that $\frac{2}{m}-x_{sn}^2=\cdots= \frac{2}{m}-x_{(s+1)n-1}^2=\frac{(2s+2)n}{m(m+n)}=\frac{2}{m+n}$ for some integer $s$. It follows that $m=(s+1)n$ which contradicts with $d>0$.
Hence, the two sets \eqref{eq:contraction-1} and \eqref{eq:contraction-2} cannot be equal if $d>0$.
Next we suppose $m=kn$. We may assume $\ket{\a_j,\b_j}=\ket{j,j+n},~\forall j$. In the following we will show there exist $x_j$'s and $y_j$'s such that $\ket{\psi}_{ABC}$ is a tripartite AME state. By Eqs. \eqref{eq:2mm+nrhob}-\eqref{eq:2mm+nrhoc}, we formulate the system of equations as follows.
\begin{equation}
\label{eq:equations}
\left\{
\begin{aligned}
&x_j^2+y_j^2= \frac{2}{m}, \quad & 0\leq j\leq m-1; \\
&x_0^2=\cdots= x_{n-1}^2=\frac{2}{m+n};   \\
&x_{n+k}^2+y_k^2=\frac{2}{m+n}, \quad & 0\leq k\leq m-n-1;  \\
&y_{m-n}^2=\cdots= y_{m-1}^2=\frac{2}{m+n}.
\end{aligned}
\right.
\end{equation} 
One can verify the system of equations \eqref{eq:equations} has the system of solutions as follows. For any $0\leq j\leq k-1$,
\begin{equation}
\label{eq:solutions}
\left\{
\begin{aligned}
&x_{jn}^2=\cdots= x_{(j+1)n-1}^2=\frac{2m-2jn}{m(m+n)}=\frac{2k-2j}{k(k+1)n},   \\
&y_{jn}^2=\cdots= y_{(j+1)n-1}^2=\frac{2(j+1)n}{m(m+n)}=\frac{2j+2}{k(k+1)n}.  
\end{aligned}
\right.
\end{equation}
Therefore, when $m=kn$, the tripartite state $\ket{\psi}_{ABC}$ in \eqref{eq:psi=2xmxn} with $\ket{\a_j,\b_j}=\ket{j,j+n},~\forall j$ and coefficients in Eq. \eqref{eq:solutions} is a tripartite AME state in $2\times m\times (m+n)$.
This completes the proof.
\end{proof}

\bibliography{ames}
\end{document}